\documentclass[a4paper,USenglish,cleveref, autoref, thm-restate]{lipics-v2021}

\pdfoutput=1 %
\hideLIPIcs  %

\title{Algorithm Engineering of SSSP With Negative Edge Weights}

\author{Alejandro Cassis}{Saarland University and Max Planck Institute for Informatics, Saarland Informatics Campus, Saarbrücken, Germany}{acassis@mpi-inf.mpg.de}{}{}
\author{Andreas Karrenbauer}{Max Planck Institute for Informatics and Saarland University, Saarland Informatics Campus, Saarbrücken, Germany}{andreas.karrenbauer@mpi-inf.mpg.de}{https://orcid.org/0000-0001-6129-3220}{}
\author{André Nusser}{Université Côte d'Azur, CNRS, Inria, France}{andre.nusser@cnrs.fr}{https://orcid.org/0000-0002-6349-869X}{This work was supported by the French government through the France 2030 investment plan managed by the National Research Agency (ANR), as part of the Initiative of Excellence of Université Côte d'Azur under reference number ANR-15-IDEX-01.}
\author{Paolo Luigi Rinaldi}{Max Planck Institute for Informatics and Saarland University, Saarland Informatics Campus, Saarbrücken, Germany}{prinaldi@mpi-inf.mpg.de}{https://orcid.org/0000-0003-1963-4516}{}

\authorrunning{A.\ Cassis, A.\ Karrenbauer, A.\ Nusser, P.L.\ Rinaldi} %



\category{} %

\relatedversion{} %

\acknowledgements{We want to thank Karl Bringmann and Danupon Nanongkai for their helpful input.}%

\nolinenumbers %

\EventEditors{John Q. Open and Joan R. Access}
\EventNoEds{2}
\EventLongTitle{42nd Conference on Very Important Topics (CVIT 2016)}
\EventShortTitle{CVIT 2016}
\EventAcronym{CVIT}
\EventYear{2016}
\EventDate{December 24--27, 2016}
\EventLocation{Little Whinging, United Kingdom}
\EventLogo{}
\SeriesVolume{42}
\ArticleNo{23}

\usepackage{algorithm}
\usepackage[noend]{algpseudocode}
\usepackage{hyperref}
\usepackage{graphicx}
\usepackage{amsmath}
\usepackage{amssymb}
\usepackage{xspace}
\usepackage{xcolor}
\usepackage[style=english]{csquotes}
\usepackage{mathtools}
\usepackage{cleveref}
\usepackage{booktabs}
\usepackage{tikz}
\usetikzlibrary{arrows.meta}
\usepackage{subcaption}
\usepackage{textcomp}

\newcommand{\GOR}{\texttt{GOR}\xspace}
\newcommand{\BFCT}{\texttt{BFCT}\xspace}
\newcommand{\BCF}{\texttt{OUR}\xspace}
\newcommand{\SSSPDag}{\textsc{FixDAGEdges}\xspace}
\newcommand{\LazyDijkstra}{\textsc{LazyDijkstra}\xspace}
\newcommand{\Decompose}{\textsc{Decompose}\xspace}

\newcommand{\bad}{\texttt{BAD}\xspace}
\newcommand{\badbfct}{\texttt{BAD-BFCT}\xspace}
\newcommand{\baddfs}{\texttt{BAD-DFS}\xspace}
\newcommand{\badgor}{\texttt{BAD-GOR}\xspace}
\newcommand{\badrdone}{\texttt{BAD-RD}\xspace}  %
\newcommand{\badrdtwo}{\texttt{BAD-RDB}\xspace}  %

\newcommand{\aug}{\texttt{AUG}\xspace}
\newcommand{\augbfct}{\texttt{AUG-BFCT}\xspace}
\newcommand{\augdfs}{\texttt{AUG-DFS}\xspace}
\newcommand{\auggor}{\texttt{AUG-GOR}\xspace}
\newcommand{\augrdone}{\texttt{AUG-RD}\xspace}  %
\newcommand{\augrdtwo}{\texttt{AUG-RDB}\xspace}  %

\newcommand{\shiftgor}{\texttt{SHIFT-GOR}\xspace}

\newcommand{\rand}{\texttt{RANDOM~RESTRICTED}\xspace}

\newcommand{\usa}[1]{\texttt{USA-#1}\xspace}

\newcommand{\BAD}{\emph{BAD}\xspace}
\newcommand{\AUG}{\emph{AUG}\xspace}

\newcommand{\USA}{\texttt{USA}\xspace}

\newcommand{\todo}[1]{}
\newcommand{\andre}[1]{}
\newcommand{\paolo}[1]{}
\newcommand{\andreas}[1]{}
\newcommand{\alejandro}[1]{}
\newcommand{\needinfo}[1]{}

\begin{document}

\maketitle

\begin{abstract}

Computing shortest paths is one of the most fundamental algorithmic graph problems.
It is known since decades that this problem can be solved in near-linear time if all weights are nonnegative.
A recent break-through by~\cite{BernsteinNW22} presented a randomized near-linear time algorithm for this problem.
A subsequent improvement in~\cite{BringmannCF23} significantly reduced the number of logarithmic factors and thereby also simplified the algorithm.
It is surprising and exciting that both of these algorithms are combinatorial and do not contain any fundamental obstacles for being practical.

We launch the, to the best of our knowledge, first extensive investigation towards a practical implementation of~\cite{BringmannCF23}.
To this end, we give an accessible overview of the algorithm and discuss what adaptions are necessary to obtain a fast algorithm in practice.
We manifest these adaptions in an efficient implementation.
We test our implementation on a benchmark data set that is adapted to be more difficult for our implementation in order to allow for a fair comparison.
As in~\cite{BringmannCF23} as well as in our implementation there are multiple parameters to tune, we empirically evaluate their effect and thereby determine the best choices.
Our implementation is then extensively compared to one of the state-of-the-art algorithms for this problem~\cite{GOR93}.
On the hardest instance type, we are faster by up to almost two orders of magnitude.
\end{abstract}

\section{Introduction}
One of the most fundamental algorithmic problems on graphs is the single-source shortest path problem: given a source vertex $s$ in the graph, the goal is to compute the distance from $s$ to all other vertices.
This is a problem with a very rich history that goes back to the early years of the field of modern computer science.
This problem has two seemingly fundamentally different variants.
Either all the edge weights are nonnegative, or we are in the general case where the edge weights are also allowed to be negative.\footnote{We indeed disregard the interesting distinction between integer and real edge weights here.}

The first variant, i.e., shortest paths on graphs with nonnegative edge weights, is very well understood.
It is known since the 60s that the celebrated Dijkstra's algorithm solves the problem in near-linear time~\cite{Dijkstra59}.
There is a long line of work (e.g., \cite{BoasKZ77, FredmanT87, AhujaMOT90, CherkasskyGS97}) giving theoretical improvements over Dijkstra's algorithm, culminating in a near-linear time algorithm with only $\log \log$ factors for graphs with integer weights~\cite{Thorup03}.
On the practical side, there also is a lot of interest in and work on the shortest path problems.
In particular, the practically efficient computation of shortest paths in road networks is very well understood, and it can be considered as one of the great success stories of algorithm engineering, see \cite{BastDGMPSWW16} for an excellent overview on this topic.

For the second variant, i.e., shortest paths with negative edge weights, progress turned out to be significantly more challenging. 
Note that as opposed to the shortest path problem with non-negative edge weights, here we also have the problem of cycles of negative total weight, namely, negative cycles.
If such a cycle is reachable from the source, we have shortest paths of arbitrary negative cost and we consequently have to detect these cases.
A classical algorithm that solves the shortest path problem with negative edge weights is the famous Bellman-Ford algorithm~\cite{ford1956, bellman:routing, moore1959shortest}, which was developed in the 50s and remained unbeaten for a long time.
However, the problem continued to be of interest and saw steady progress on the theoretical side, e.g. \cite{Goldberg95, Sankowski05, KleinRRS94}.
Given its fundamental nature, the problem also drew interest from a practical perspective. We refer the reader to~\cite{Pape74, DialGKK79, Pallottino84, GOR93} for some early experimental work. Some more recent works focused on negative cycle detection, also called the shortest-path feasability problem, see e.g.~\cite{CherkasskyG99, ChandrachoodanBL01, CherkasskyGGTW09}.
We refer to \cite{CherkasskyGGTW09} as an excellent resource on practical shortest path algorithms with negative edge weights and negative cycle detection.
Note that shortest path algorithms can also be used as a subroutine for solving min-cost flow problems with the so-called successive shortest path approach~\cite{EdmondsK72}.

Despite the progress in shortest paths with negative edge weights, until very recently all known algorithms were still far from achieving near-linear running time.
In 2022, this barrier was overcome by~\cite{ChenKLPGS22}, who gave an almost-linear time algorithm for the more general min-cost flow problem.
In an independent work, a near-linear time algorithm for shortest paths with negative weights was discovered by~\cite{BernsteinNW22}.
Both of these results constitute stunning theoretical breakthroughs, see also~\cite{CACM_article}.
One disadvantage of the min-cost flow algorithm of~\cite{ChenKLPGS22} is that it is very complex and relies on intricate tools from continuous optimization, which suggests that it is still far from being of practical use.\footnote{There was a very recent attempt \cite{kavi2024partialimplementationmaxflow} to implement the algorithm of \cite{ChenKLPGS22}.
This implementation is only partial, suggesting that indeed it is difficult and the practicability of the theoretical algorithm remains unclear.}
On the other hand, a notable feature of~\cite{BernsteinNW22} is that it does not rely on complex technology.
Indeed, the algorithm is purely combinatorial (it uses tools similar to classical graph decomposition techniques that are already known since the 80s, see e.g. \cite{AwerbuchGLP89}), and none of the techniques that it uses renders it a priori impractical.
This gave rise to the hope of practical near-linear time algorithms for shortest paths with negative edge weights.
This hope was further boosted by a result that reduced the number of logarithmic factors in the running time from 9 to 3 and thereby also resulted in a significantly simpler algorithm~\cite{BringmannCF23}.\footnote{We note that the algorithm was simplified, but the analysis is still challenging and out of the scope of this work to explain.}
We note that in~\cite{BringmannCF23}, the authors explicitly mention the goal of paving the way for a \enquote{comparably fast implementation}.

\subsection*{Our Contribution}

We initiate the transfer of the theoretical break-through on the problem of shortest paths with negative edge weights \cite{BernsteinNW22, BringmannCF23} to the application domain and provide the first implementation inspired by the work of \cite{BringmannCF23}, to the best of our knowledge.
To that end, we first present an overview of the algorithm of \cite{BringmannCF23} to make it easily accessible to a broader audience and especially practitioners.\footnote{Naturally, explaining the running time analysis of the algorithm of \cite{BringmannCF23} is out of the scope of this paper and we merely focus on the algorithm.}
We discuss different adaptions that are necessary to obtain an implementation that is competitive with existing algorithms.
We also show an improved instance-based upper bound on a crucial parameter that controls the recursion depth of the algorithm.
We create a set of benchmark instances that is inspired by the work of \cite{CherkasskyGGTW09}, adapting their benchmark instances to make them more difficult to solve for our implementation in order to allow for a fair comparison.
As the algorithm of \cite{BringmannCF23} contains multiple constants that can be chosen in different ways and our modifications also allow for different parameter settings, we conduct experiments to empirically determine a good parameter choice.
We then conduct extensive experiments to understand the running time and especially also the scaling behavior of our implementation compared to one of the state-of-the-art algorithms, namely \GOR~\cite{GOR93}.
On the hardest instance type, we are faster by up to almost two orders of magnitude, suggesting that our implementation particularly thrives on hard instances.
Indeed, the experiments also suggest that the running time of our implementation scales near-linearly.

We note that Li and Mowry~\cite{DBLP:journals/corr/abs-2411-19449} very recently and independent of our work further simplify the algorithm of~\cite{BringmannCF23}. However, they do not provide an implementation and conduct no experimental evaluation.

\section{Preliminaries} \label{sec:prelims}

We first introduce notation, then explain the type of graphs that we are working on, and finally explain the algorithm of~\cite{BringmannCF23} that we build our engineered solution on.

\subsection{Setting} \label{sec:setting}
In this work, we are always given a directed, weighted graph $G = (V,E,w)$ equipped with a weight function $w: E \mapsto \mathbb{Z}$ that assigns (potentially negative) integer weights to the edges.
With slight abuse of notation, we use $|G|$ to denote the number of vertices of a graph $G$.
Given a subset of vertices $C \subset V$, we denote with $G[C]$ the induced subgraph of $C$ in $G$.
We sometimes consider the graph $G_{\geq 0}$, which we define as the graph $G$ with the modified cost function $w_{\geq 0}(e) = \max\{w(e), 0\}$ where all negative edge weights are set to zero.
We also need access to the reverse graph.
We hence denote with $G^\mathrm{out}$ the original graph $G$, and with $G^\mathrm{in}$ the graph $G$ but where each edge is flipped in its orientation.
We use the notation $B^\mathrm{dir}_{\geq 0}(v,r)$ to denote the set of vertices that have distance at most $r$ from $v$ in $G^\mathrm{dir}_{\geq 0}$, for $\mathrm{dir} \in \{\mathrm{in}, \mathrm{out}\}$.
Given such a ball $B \coloneqq B^\mathrm{dir}_{\geq 0}(v,r)$, we denote its boundary by $\partial B$, which comprises of all the edges that lead from nodes inside the ball to nodes outside the ball.
Finally, given a direction $\mathrm{dir} \in \{\mathrm{in}, \mathrm{out}\}$, we denote by $\overline{\mathrm{dir}}$ the opposite direction.

The goal of this work is to solve the following problem efficiently in practice:
\begin{definition}[SSSP with Negative Weights]
Given a directed graph $G = (V,E,w)$ equipped with a weight function $w: E \mapsto \mathbb{Z}$ and a source vertex $s \in V$, compute the shortest path costs from $s$ to all vertices in $V$ w.r.t.\ the weights $w$ or assert that $G$ contains a negative cycle.
\end{definition}
SSSP with Negative Weights reduces to finding a \emph{potential function} $\phi: V \mapsto \mathbb{Z}$ such that the modified weights $w(u,v) + \phi(u) - \phi(v)$ of all edges $(u,v) \in E$ are nonnegative; we call such a potential function \emph{valid}.
This is also called \enquote{Johnson's Trick}~\cite{Johnson77}.
There are three important facts about such potential functions:
\begin{itemize}
    \item a valid potential function $\phi$ always exists iff $G$ does not contain a negative cycle, and
    \item $\phi$ preserves the structure of the shortest paths in $G$, and
    \item given shortest path costs w.r.t.\ $\phi$, we can easily reconstruct the original costs.
\end{itemize}
Especially, if we found a valid potential function $\phi$, we can just apply Dijkstra's algorithm~\cite{Dijkstra59} on $G$ with the modified weights and hence obtain our shortest path distances in the graph with the original weights.

\subsection{Restricted Graphs}

Consider any cycle $C = e_1, \dots, e_\ell$ in $G$, its \emph{mean} is defined as $\sum_{i \in [\ell]} w(e_i) / \ell$.
The \emph{minimum cycle mean} is defined as the minimum mean any cycle in $G$ has.
In this work we focus on restricted graphs.
\begin{definition}
    A graph $G = (V, E, w)$ is restricted if $w(e) \geq -1$ for all $e \in E$, and the minimum cycle mean of $G$ is at least 1.
\end{definition}
Note that this definition also implies that $G$ has no negative cycle.
Furthermore, the original definition also includes the existence of a super source, due to technical reasons.
We drop this requirement.

We mostly restrict to these types of graphs in this work (unless mentioned otherwise), due to the following reasons:
\begin{itemize}
    \item 
    The main contribution of~\cite{BringmannCF23} is an algorithm that runs on restricted graphs. How to generalize such an algorithm to the general case is technical, but comparably simple, and was already previously known.
    Hence, restricting to these types of graphs enables for a cleaner evaluation of how their techniques perform in practice.
    \item Most of our benchmark graphs are derived from the hard instances described in~\cite{CherkasskyGGTW09}.
    These graphs are restricted or almost restricted and thus, even if we adapt our implementation for the general case, we expected the performance on these graphs to not significantly change.
    \item The restriction does not have any effect on the correctness of our implementation: even if our implementation is given a non-restricted graph, it still outputs the correct result.\footnote{This is the case as Line~\ref{l:lazy_dijkstra_call} in Algorithm~\ref{alg:restricted_bcf} computes the correct result also on non-restricted graphs.}
    Our implementation is merely not optimized for these cases and might have higher running times.
    In particular, if the graph contains a negative cycle, it still will eventually be detected by our implementation.
\end{itemize}

While our implementation can easily be adapted to general negative-cycle-free graphs by the scaling technique of~\cite{BringmannCF23}, we note that the negative cycle detection of~\cite{BringmannCF23} is impractical as it works via budgeting --- a technique that, very roughly speaking, given the theoretical worst-case running time of the algorithm on well-formed inputs, will stop if it is exceeded.
We discuss negative cycles more in Section~\ref{sec:implementation}.

\subsection{Near-Linear Time Algorithm} \label{sec:bcf}
In this subsection we give a description of~\cite[Algorithm~2]{BringmannCF23}.
For the readers' convenience, we refer to the statements and algorithms in the freely available full version~\cite{BCF23-arxiv} of~\cite{BringmannCF23}.
Note that in the algorithm of~\cite{BringmannCF23} a designated source $s \in V$ is added and connected to all vertices via zero-weight edges.
In our description and our implementation, we replace the explicit source node by an implicit source.
Computing a shortest path tree from this implicit source then gives us a valid potential function, and we can subsequently run Dijkstra's algorithm from any source with the modified weights, see \Cref{sec:setting}.

\subsubsection{Subroutines}
Before describing the main algorithm, we have to describe several important subroutines.
Note that most of these subroutines get a potential function $\phi$ as input.
This is due to the fact that the main algorithm recursively creates new potential functions, fixing more and more of the graph until we found a valid one.

\paragraph*{\LazyDijkstra}
We first describe an algorithm that is efficient in the case where each path in the shortest path tree contains only few negative edges.
This subroutine is called \LazyDijkstra, see also \cite[Lemma 3.3]{BernsteinNW22} and \cite[Lemma 25]{BCF23-arxiv}.
The algorithm intuitively is Dijkstra's algorithm interleaved with Bellman Ford edge relaxations of all the relevant negative edges, see \Cref{alg:lazy_dijkstra}.

\begin{algorithm}[t]
\caption{\LazyDijkstra} \label{alg:lazy_dijkstra}
\begin{algorithmic}[1]
\Procedure{\LazyDijkstra}{$G, \phi$}
\State $d \gets$ distance array used in the following
\State \textbf{for each} $v \in V$ \textbf{do} $d(v) \gets -\phi(v)$
\State $Q \gets$ priority queue containing each $v \in V$ with its initial distance $-\phi(v)$
\While{$Q \neq \emptyset$}
\State Dijkstra's alg.\ on edges in $E$ with non-negative cost w.r.t.\ $\phi$, updating $d$ and $Q$
\State $A \gets$ vertices that were improved (or newly initialized)
\State One round of Bellman-Ford from $A$ using $\phi$ and updating $d$
\State $Q \gets$ vertices that were improved
\EndWhile
\State \Return $\phi + d$ %
\EndProcedure
\end{algorithmic}
\end{algorithm}

\paragraph*{\SSSPDag}
Now we describe the subroutine \SSSPDag, see also \cite[Lemma 3.2]{BernsteinNW22} and \cite[Lemma 26]{BCF23-arxiv}, see Algorithm~\ref{alg:fix_dag_edges}.
This subroutine takes a set of strongly connected components $C_1, \dots, C_\ell$ in their topological ordering according to $G$ as input, and we are guaranteed that for each $i \in [\ell]$ the edges within $C_i$ are non-negative with respect to $\phi$.
This subroutine finds a new potential function such that all edges become non-negative.
The idea to achieve this is simple: If we contract each $C_i$ to a single vertex in $G$, then we obtain a DAG.
Hence, we can simply assign large negative potentials that decrease (i.e., become more negative) the later the component is in the topological ordering.

\begin{algorithm}[t]
\caption{\SSSPDag} \label{alg:fix_dag_edges}
\begin{algorithmic}[1]
\Require{$C_1, \dots, C_\ell$ given in topological ordering, and edge costs are non-negative within all $C_i$ w.r.t. $\phi$}
\Procedure{\SSSPDag}{$G, C_1, \dots, C_\ell, \phi$}
    \State $M \gets \min_{(u,v) \in E} \{w(u,v) + \phi(u) - \phi(v), 0\} - 1$
    \For{$i \gets 1, \dots, \ell$}
        \State \textbf{for} $v \in C_i$ \textbf{do} $\phi(v) \gets \phi(v) + i \cdot M$
    \EndFor
    \State \Return $\phi$
\EndProcedure
\end{algorithmic}
\end{algorithm}

\paragraph*{\Decompose} \label{parag:decompose}
This routine is at the heart of the near-linear time algorithm by~\cite{BCF23-arxiv}, see~\cite[Section 3.1]{BCF23-arxiv}.
We provide the pseudo-code in Algorithm~\ref{alg:decompose}.
At a high level, the goal of \Decompose is to decompose the graph into small pieces by removing a small number of cut edges from the graph.
The decomposition happens based on cutting out balls whose radii are expected to be in the order of~$\kappa/\log n$, where $\kappa$ is a parameter that we explain and set in Section~\ref{subsubsec:alg}.
To achieve the goal, we want to grow the balls around vertices whose vicinity is relatively small.
These are so-called light vertices, whose $\frac{\kappa}{4}$-ball contains at most a $\frac{3}{4}$-fraction of the vertices of $G$.
However, classifying each vertex exactly is too costly, so instead we estimate the set of light vertices.
This is the idea behind the set of vertices $L$ that we compute in Lines~\ref{l:start_light}~to~\ref{l:end_light} of \Cref{alg:decompose}.
To this end, we randomly sample vertices, growing a ball of radius $\frac{\kappa}{4}$ around them in opposite direction, and marking all the vertices that are contained in this ball.
This is done $k$ times, where we choose $k$ to be $50 \log n$.
The set $L$ then consists of the vertices that were marked less than $\frac{3}{5}k$ times.

Note that this subroutine introduces randomization into the algorithm as the radius of the balls is sampled from a geometric distribution, where $\mathrm{Geom}(p)$ denotes the geometric distribution with mean $1/p$.
Furthermore, in our description of \Decompose, we use the constant factors of our implementation instead of the ones from~\cite{BringmannCF23}.

\begin{algorithm}[t]
\caption{\Decompose} \label{alg:decompose}
\begin{algorithmic}[1]
\Procedure{\Decompose}{$G, \kappa$}
\State $S \gets \emptyset$
\For{$\mathrm{dir} \in \{\text{in}, \text{out}\}$}
\State $k \gets \lceil 50 \log |G| \rceil$ \label{l:start_light}
\For{$i = 1, \dots, k$}
    \State $v \gets $ random vertex from $G$
    \State Mark all vertices in $B^{\overline{dir}}_{\geq 0}(v, \kappa/4)$
\EndFor
\State $L \gets$ vertices marked less than $\frac{3}{5}k$ times \label{l:end_light}
\While{there is a $v \in L$}
\State $B \gets B_{\geq 0}^{\mathrm{dir}}(v, \mathrm{Geom}(20 \log(|G|) / \kappa))$
\State $S \gets S \cup \partial B$, $L \gets L \setminus B$, and $G \gets G \setminus B$
\EndWhile
\EndFor
\State $C_1, \dots, C_\ell \gets$ topol.\ sorted SCCs of $G \setminus S$
\State \Return $C_1, \dots, C_\ell, S$
\EndProcedure
\end{algorithmic}
\end{algorithm}

\begin{algorithm}[t]
\caption{Algorithm for restricted graphs from~\cite{BringmannCF23}} \label{alg:restricted_bcf}
\begin{algorithmic}[1]
\Procedure{RestrictedSSSP}{$G, \phi, \kappa$}
\If{$\kappa \leq 2$}
    \State \Return \Call{\LazyDijkstra}{$G, \phi$}
\EndIf
\State $C_1, \dots, C_\ell, S \gets$ \Call{\Decompose}{$G, \kappa$}
\For{each component $C_i$}
\State \textbf{if} $|C_i| \geq \frac{3}{4} |G|$ \textbf{then} $\kappa_i \gets \kappa/2$ \textbf{else} $\kappa_i \gets \kappa$ \label{l:progress}
\State $\phi \gets \text{\Call{RestrictedSSSP}{$G[C_i], \phi, \kappa_i$}}$
\EndFor
\State $\phi \gets \text{\Call{\SSSPDag}{$G \setminus S, C_1, \dots, C_\ell, \phi$}}$
\State $\phi \gets \text{\Call{\LazyDijkstra}{$G, \phi$}}$ \label{l:lazy_dijkstra_call}
\State \Return $\phi$
\EndProcedure
\end{algorithmic}
\end{algorithm}

\subsubsection{Algorithm} \label{subsubsec:alg}
We can now describe the main algorithm on restricted graphs from~\cite{BringmannCF23}.
See \Cref{alg:restricted_bcf} for the pseudo code.
Note that we use $|G|$ here to denote the number of vertices of the graph $G$, and also that the algorithm is recursive and calls itself on the graphs of components induced by the separators found by Algorithm~\ref{alg:decompose}.

We focus here on the intuitive understanding of the described algorithm; for the in-depth description, see~\cite[Section 3]{BCF23-arxiv}.
The algorithm is recursive and we make progress in recursive calls by either reducing the component size or by halving the parameter $\kappa$, see Line~\ref{l:progress} in \Cref{alg:restricted_bcf}, the base case being $\kappa \leq 2$.
We shed light on the role of $\kappa$ in a moment.
However, first consider a single recursive step.
We can categorize the negative edges into three classes after the \Decompose call: (i) edges within strongly connected components $C_i$, (ii) cut edges $S$, and (iii) other edges between the strongly connected components.
Type (i) edges are handled recursively.
Type (ii) edges could be hard to handle, but it follows from the theoretical analysis that they are expected to be few, thus, we expect to find a valid potential function for them efficiently using \LazyDijkstra. 
Type (iii) edges are DAG edges when considering the components as being contracted into a vertex. Hence, we easily find a valid potential function for them using the \SSSPDag subroutine.

Let us now explain the crucial role of the parameter~$\kappa$.
The value $\kappa(G)$ of a given restricted graph $G$ is the largest number of negative edges in any simple path with non-positive cost.
This intuitively measures the complexity of the shortest path problem incurred by the negative edges.
The parameter $\kappa$ is supposed to be an upper bound on $\kappa(G)$.
If $\kappa(G)$ is small, then we can efficiently solve the shortest path problem using \LazyDijkstra.
This is why halving $\kappa$ intuitively makes progress --- we expect to be in a less complex case.

\section{Implementation} \label{sec:implementation}

We now describe our implementation and different algorithmic choices.

\subsection{Our Algorithm} \label{subsec:ouralg}
As the goal of this work is to understand the practicability of the algorithmic approach of~\cite{BringmannCF23}~and~\cite{BernsteinNW22}, we use the structure described in \Cref{sec:bcf}, and in particular \Cref{alg:restricted_bcf}, as basis and modify it to obtain a practically faster algorithm.

As a first preprocessing step, we decompose the graph into strongly connected components.
This allows us to run our algorithm separately on each component and later obtain a valid potential function that also makes the edges in-between the components positive using \SSSPDag.
To this end, we also have to choose one of the many algorithms to compute strongly connected components.
We use Kosaraju's and Sharir's algorithm~\cite{Sharir81}, which also gives us a topological sorting of the strongly connected components as a by-product.

A subroutine that is pervasive in \cite{BringmannCF23} is Dijkstra's algorithm.
In~\cite{BringmannCF23} the authors use Dijkstra's algorithm with a priority queue by Thorup~\cite{Thorup03} with the goal to shave an additional logarithmic factor.
We instead use a classical priority queue, more precisely, we use a 4-heap that has proven to be fast in practice in shortest path problems with non-negative edge weights, see e.g. \cite{GoldbergKW06}.

A critical parameter that controls the size of the components that we recurse on and thereby also the recursion depth is $\kappa$ --- hence choosing $\kappa$ well can have significant influence on the practical running time.
In Section~\ref{sec:theory}, we show that using $\mathrm{diam}(G_{\geq 0})$ is an upper bound to $\kappa$ on any restricted graph $G$.
As $\mathrm{diam}(G_{\geq 0})$ could be greater than $n$, we set $\kappa = \min\{n, \mathrm{diam}(G_{\geq 0})\}$ in our implementation.
It is also possible that the graph contains only very few negative edges in total, so we could additionally have the total number of negative edges in $G$ as third argument in the $\min$ above.
However, as all the instances that we run our experiments on contain a significant amount of negative edges, we do not use this optimization in the experiments.

The computation of the estimated set of light vertices $L$ in \Decompose is an important factor in the algorithm.
To reduce the running time, we can reduce the number of randomly sampled vertices that we use to estimate the lightness.
While the choice in the work of~\cite{BringmannCF23} is aimed at ensuring a high probability of success, it is conceivable that on practical instances we need less vertices for a good estimation.
To this end, we introduce a parameter $K \in \mathbb{N}$ to our algorithm such that we choose $1/K$-times as many vertices to estimate the lightness in the \Decompose routine.

We also have to decide at which point we want to invoke the base case and call \LazyDijkstra.
In \cite{BringmannCF23} this is done when $\kappa \leq 2$, see \Cref{alg:restricted_bcf}.
While this is theoretically guaranteed to occur within a logarithmic number of recursive calls with sufficiently high probability, in practice this can slow down the implementation significantly.
Hence, we settled on calling the base case if $n + \kappa \leq 300$.
We use the sum as we want to only call the base case if both $\kappa$ and $n$ are small.
Note that due to how we set $\kappa$, we also always have $\kappa \leq n$.

Finally, our implementation also performs negative cycle detection, however, this is not its strong suit.
In~\cite{BringmannCF23} negative cycle detection is done via budgeting, i.e., when the algorithm exceeds a specific running time, then the existence of a negative cycle is reported; this technique is impractical and we do not use it.
For an alternative approach, note that the only cases where negative cycles can occur in \Cref{alg:restricted_bcf} are the \LazyDijkstra calls.
Hence, we add negative cycle detection to \LazyDijkstra by merely checking the number of distance improvements that any vertex sees.
If this number exceeds the number of vertices, then we detected the existence of a negative cycle and we report this.
How to adapt our implementation to also perform efficient negative cycle detection is an intriguing task for future work.\footnote{We note that in a very recent parallel and independent work that appeared on arXiv, an improved cycle detection is presented, but this work does not present any implementation or experiments~\cite{DBLP:journals/corr/abs-2411-19449}.} 

Let us consider what asymptotic running time we can hope for for our implementation.
Note that we use the wording \enquote{hope for} on purpose, as the above modifications can potentially break the theoretical guarantees and our statements should be considered an intuitive sketch.
The theoretically expected running time for the version of \Cref{alg:restricted_bcf} in \cite{BringmannCF23} is $O((m + n \log \log n) \log^2 n)$, where $n$ and $m$ are the number of vertices and edges in the input graph, respectively.
We perform one crucial modification, that is, we use a $4$-heap instead of the priority queue by~\cite{Thorup03}.
This replaces the $\log \log n$ factor in the above running time by a $\log n$ factor.
Hence, we can hope for an $O((m+n \log n) \log^2 n)$ running time.

\subsection{Other Algorithmic Choices} \label{subsec:other}

In our implementation, following~\cite{BringmannCF23}, we use \LazyDijkstra as algorithm to compute a valid potential function on graphs where each shortest path should only contain a small number of negative edges.
This happens in the base case of \Cref{alg:restricted_bcf} as well as in the case of obtaining a valid potential function for the cut edges in-between components.
Instead of \LazyDijkstra we could actually employ any other shortest path algorithm for graphs with negative edge weights that is fast in practice.
Especially, we can also use the baseline algorithm that we later compare to in Section~\ref{sec:experiments} as subroutine.

\section{Theoretical Insights} \label{sec:theory}

In \Cref{alg:restricted_bcf}, we either make progress by reducing the component size or by reducing $\kappa$. In a restricted graph we know that $\kappa \leq n$ as this is the maximal length of a simple path. This is sufficient in theory, but in practice we prefer a tighter and instance-dependent upper bound on $\kappa$. The following lemma establishes an upper bound on $\kappa$ that comes from the diameter of $G_{\geq 0}$.

\begin{lemma}[Diameter upper bound] \label{lem:diam_ub}
    Given a restricted graph $G$ that consists of a single strongly-connected component, we have that $\kappa(G) \leq \mathrm{diam}(G_{\geq 0})$.
\end{lemma}

\begin{proof}
We denote by $w(P,G) \coloneqq \sum_{e \in P} w_G(e)$ the cost of the path $P$ in graph $G$.
Let $P$ be the path from $u$ to $v$ in $G$ that realizes $\kappa(G)$.
Hence, $P$ has $\kappa(G)$ negative edges and $w(P,G) \leq 0$.
Let $P'$ be the shortest path from $v$ to $u$ in $G_{\geq 0}$.
Due to $G$ being restricted (and hence having minimum cycle mean at least 1) and as $P$ has $\kappa(G)$ negative edges, we know that
\[
    w(P \circ P', G) \geq |P| + |P'| \geq \kappa(G).
\]
Using the above statement and that $w(P,G) \leq 0$, we obtain that the diameter of $G_{\geq 0}$ is at least
\[
w(P', G_{\geq 0}) \geq w(P', G) \geq w(P \circ P', G) \geq \kappa(G),
\]
hence $\kappa(G) \leq \mathrm{diam}(G_{\geq 0})$.
\end{proof}

Note that it suffices to upper bound $\kappa(G)$ on strongly connected components as we only ever use Algorithm~\ref{alg:restricted_bcf} on strongly connected graphs, see \Cref{sec:implementation}.

\section{Experiments} \label{sec:experiments}

We now experimentally evaluate our implementation.
To that end, we had the following challenges to solve:
\begin{itemize}
    \item What algorithm(s) do we compare to? (Section~\ref{sec:competitor_algorithm})
    \item What data do we perform experiments on? (Section~\ref{sec:data})
    \item What do we parameterize in our algorithm and how do we choose these parameters? (Section~\ref{subsec:parameters})
\end{itemize}
In the following, we refer to our implementation as \BCF.
We conduct each experiment 5 times and present the average and the standard error of the mean, unless noted otherwise.
For the quantitative evaluation of our experimental data, we use regression to the model $a \cdot m^b$ by applying the common log-log-transform before an ordinary least squares fit from \texttt{statsmodels (v.0.14.2)} in Python. We report the errors of the slopes with a confidence interval of $95\%$. The error bars of single data points are reported with $1$ standard error of the mean.

\subsection{Code \& Hardware}

We wrote the entire experimental framework and all the algorithms in C++20. We compiled our code using GCC on Debian.
Our implementation is publicly available\footnote{\url{https://anonymous.4open.science/r/negative-weight-shortest-path}
} and it will remain public (in deanonymized form) in case of acceptance of the paper.
We ran the experiments on a server with 48 Intel Xeon E5-2680 v3@2.50GHz CPUs with 256 GB of RAM.
We note that we did not parallelize our implementation as our comparison is with a single-threaded implementation, but parallelization of our implementation is straight-forward (the recursive calls are independent).

\subsection{Baseline Algorithm} \label{sec:competitor_algorithm}
We compare our implementation against a state-of-the-art algorithm by Goldberg and Radzik~\cite{GOR93} called \GOR.
More specifically, we use the implementation in SPLIB\footnote{See \url{https://www3.cs.stonybrook.edu/~algorith/implement/goldberg/distrib/splib.tar}. We modified it such that we can use it with our graph data structure and compile it within our implementation. We ensured that this did not incur a significant slowdown.}, which was developed by Cherkassky, Goldberg, and Radzik.
We also implemented the algorithm called \BFCT from~\cite{CherkasskyGGTW09} (a previous version was described in \cite{CherkasskyG99}). However, as our implementation was consistently outperformed by \GOR, we only compare to \GOR here.
Our reason for choosing \GOR is twofold:
\begin{itemize}
    \item In~\cite{CherkasskyGGTW09}, the authors show that \GOR has overall the fastest average running times for scans of a single vertex and a good performance on all instances regarding the number of scans required to compute a shortest path tree, without large deviations.
    \item A highly-tuned implementation of \GOR to compare to is publicly available.
\end{itemize}

\subsection{Data} \label{sec:data}
In the following, we present a brief overview of all the different types of instances we perform experiments on. For completeness, we give the full details in \cref{sec:app_instances}.

\subparagraph*{\bad instances from~\cite{CherkasskyGGTW09}} \label{sec:bad}
These instances are worst-case instances that are take from~\cite{CherkasskyGGTW09} and presented in \cref{sec:app_bad_instances}: \badbfct, \baddfs, \badgor, \badrdone, and \badrdtwo
.
We take these instances as each of them is constructed to be pathological for a specific algorithm for shortest paths with negative edge weights. They are all DAGs, and the algorithms that they are adversarial constructions for are different variants of Bellman-Ford.

\subparagraph*{Augmented \bad instances} \label{sec:bad_aug}
As all of the \bad instances are DAGs, and our implementation specifically exploits DAG structure\footnote{We first decompose the graph into strongly connected components, which are singletons in the case of a DAG, and we then call \SSSPDag on the set of components, which is a simple linear-time algorithm. Hence, our implementation degenerates to a simple algorithm to solve the shortest path problem on DAGs with negative edge weights.}, a comparison merely on the \bad instances would not be reasonable and fair.
However, these instances still reveal an interesting structure and the work of~\cite{CherkasskyGGTW09} is the most comprehensive overview of shortest path algorithms on graphs with negative edge weights that we are aware of.
Hence, we still use these instances, but to ensure that these instances are also interesting and difficult for our implementation, we perform two crucial modifications (denoting these instances by pre-fixing them with \texttt{AUG}, e.g., \auggor).
First, we randomly permute the vertex labels. Second, we augment the graphs with 5 times as many (random) edges as the original graph, not introducing multi-edges, and assigning them a large weight to not change the shortest path structure.
These modifications ensure that the minimum cycle mean is at least 1, and all instances except the modified \badgor also fulfill the property that all edge weights are at least~-1.

\subparagraph*{\shiftgor}
These instances result from studying the behavior of \BCF with \GOR used as subroutine instead of \LazyDijkstra (see Section~\ref{subsec:other}) on the \auggor instances.
Interestingly, we found that the potential shifts that our implementation creates (after the \SSSPDag step) greatly deteriorates the performance of \GOR, creating a novel, non-trivial, very hard instance for \GOR.

\subparagraph*{\rand instances} We use these instances to analyze the behavior of \BCF on random restricted graphs.
How to generate such instances is not obvious as we want the graph to have edges with weight $-1$ to make it interesting, while ensuring a minimum cycle mean of at least $1$.
To this end, we create a random graph with valid minimum cycle mean and then iteratively perform potential shifts on shortest path trees, to finally do a weight shift to create negative edges.

\subparagraph*{\USA instances}
We also test our implementation on the \enquote{Full USA} distance graph instance from the 9th DIMACS challenge\footnote{\url{http://www.diag.uniroma1.it/~challenge9/}}.
To obtain a graph with negative edge weights for our experiments, we perform a potential shift using a shortest path tree and random potential shifts at vertices within the range $[0, W]$ for some $W$.
We thereby obtain edge weights which are at least $-W$.

\subsection{Parameters} \label{subsec:parameters}

As discussed in \Cref{sec:implementation}, there are several choices to be made and parameters to be set.
We list the interesting and non-obvious choices and parameters here and empirically evaluate them to understand what good choices are.
Note that we set and evaluated the parameters in the order described here, i.e., the best choice of the first mentioned parameter is already used in the subsequent experiments.
Hence, we chose the order below carefully and on purpose.

As described in Section~\ref{subsec:ouralg}, we can change the number of random vertices that we use to estimate the lightness in \Cref{alg:decompose}.
To that end, our implementation has a parameter $K \in [1, \infty]$ such that we sample $\max\{1, k/K\}$ random vertices, where $k$ is the theoretical number of vertices as used in \Cref{alg:decompose}.
We tested \BCF with $K \in \{1, 5, 10, 20, 40, \infty\}$ on several large instances, see Table~\ref{tab:app_k-factor}.
The speed-up induced by this parameter is significant; it is up to almost a factor of 4.
We notice that the worst choice is $K = 1$, i.e., the original value that is used in \Cref{alg:restricted_bcf}.
There is no clear best choice between $K$ being $20$, $40$, or $\infty$.
To strike a balance, we use $K = 40$ in our subsequent experiments.

Next we consider whether we should use the improved upper bound on $\kappa$ in our implementation, see Section~\ref{sec:theory}~and~Section~\ref{subsec:ouralg}.
We call this method the \emph{diameter upper bound} here.
In almost all of our experiments we noticed a similar scaling behavior with the diameter upper bound and without it.
On one hand, on small instances computing the diameter sometimes poses a significant running time overhead and makes it a bit slower in total.
On the other hand, on the \rand instances, using the diameter upper bound leads to a better scaling behavior and a speed-up of almost an order of magnitude, see Figure~\ref{fig:app_diam-apprx-rand}.
Hence, despite small disadvantages on some instances, we decide to use it as it can massively speed-up some cases.

Finally, as discussed in Section~\ref{subsec:ouralg}, we can also replace \LazyDijkstra by any other algorithm that performs shortest path computations on graphs with negative edge weights.
The obvious replacement to evaluate is \GOR.
While the speed-up is minor when using \GOR as replacement, the slowdown is massive on \auggor instances, see Figure~\ref{fig:app_use-lazy-gor}.
Hence, we conclude that using \LazyDijkstra as base algorithm is the more robust choice.

\begin{table*}[h!]
    \caption{$K$ on different instances with $2 \cdot 10^{7}$ edges. Time is expressed in seconds and rounded to full seconds.
    }
    \centering
\begin{tabular}{ccccccc}
\toprule
          &      $K = 1$ &      $K = 5$ &     $K = 10$ &            $K = 20$ &             $K = 40$ &        $K = \infty$ \\
\midrule
  \baddfs &  $124 \pm 5$ &   $48 \pm 3$ &   $56 \pm 2$ &          $39 \pm 2$ &  $\mathbf{35 \pm 1}$ &          $38 \pm 2$ \\
 \badbfct &  $168 \pm 6$ &   $62 \pm 1$ &   $57 \pm 4$ & $\mathbf{48 \pm 1}$ &           $52 \pm 6$ &          $50 \pm 5$ \\
\badrdone &  $164 \pm 5$ &   $62 \pm 1$ &   $64 \pm 2$ &          $48 \pm 2$ &           $49 \pm 3$ & $\mathbf{43 \pm 2}$ \\
\badrdtwo &  $126 \pm 4$ &   $51 \pm 3$ &   $47 \pm 3$ &          $40 \pm 1$ &  $\mathbf{37 \pm 1}$ & $\mathbf{37 \pm 1}$ \\
  \badgor & $260 \pm 20$ & $270 \pm 20$ & $250 \pm 30$ &        $230 \pm 30$ & $\mathbf{203 \pm 3}$ &        $240 \pm 20$ \\
\bottomrule
\end{tabular}
    \label{tab:app_k-factor}
\end{table*}

\begin{figure}[h!]
\begin{minipage}{.47\textwidth}
    \centering
    \includegraphics[width=\linewidth]{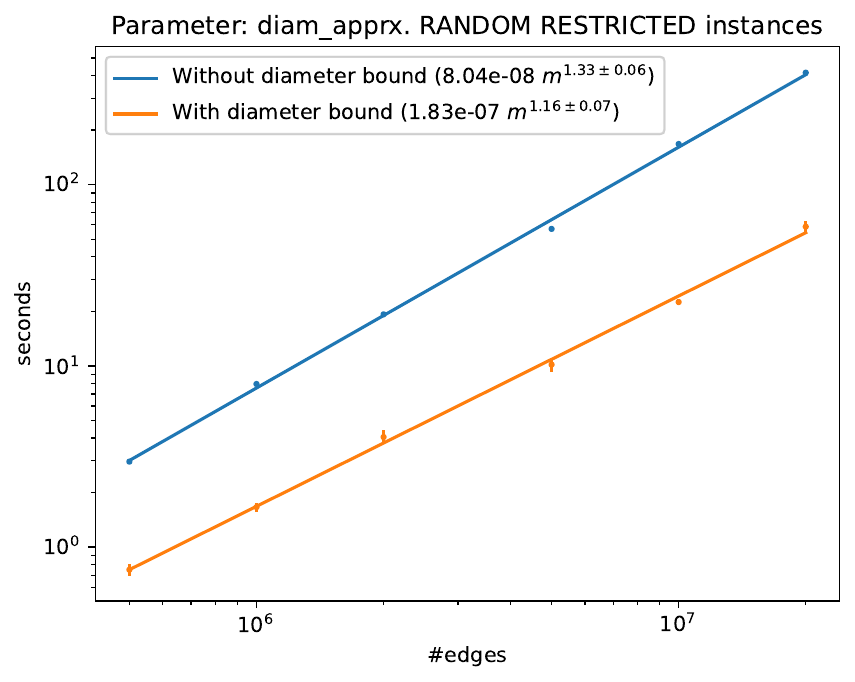}
    \caption{Performance of \BCF with and without the diameter upper bound on restricted random instances (see Section~\ref{sec:data}). \paolo{Is it fine to leave here diam apprx as parameter in the title (it is what we use in the code)? I can easily regenerate these pictures.} \andre{Doesn't matter as it's only appendix.}}
    \label{fig:app_diam-apprx-rand}
\end{minipage}
\hfill
\begin{minipage}{.47\textwidth}
    \centering
    \includegraphics[width=\linewidth]{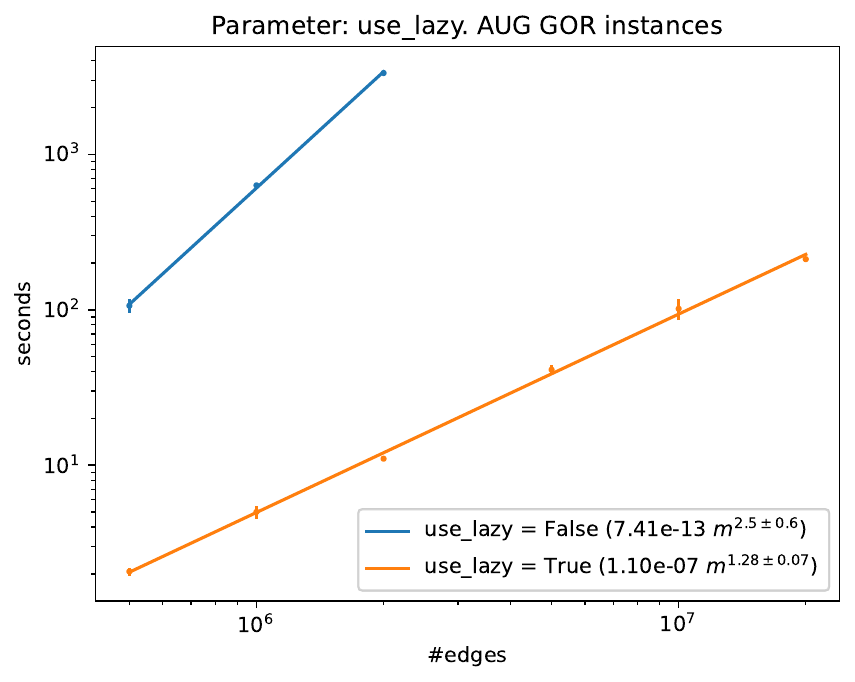}
    \caption{Performance of \BCF with \LazyDijkstra and \GOR on \auggor (see Section~\ref{sec:data}). This plot is the reason why we investigated on \shiftgor. \paolo{Should I say ``Use LazyD'' and ``Use GOR'' instead? If yes, how should I call the parameter in the title of this picture? Simply ``Subroutine''} \andre{Doesn't matter as it's only appendix.}}
    \label{fig:app_use-lazy-gor}
\end{minipage}
\end{figure}

\begin{figure*}[t!]
    \centering
    \begin{subfigure}[t]{.47\linewidth}
    \includegraphics[width=\linewidth]{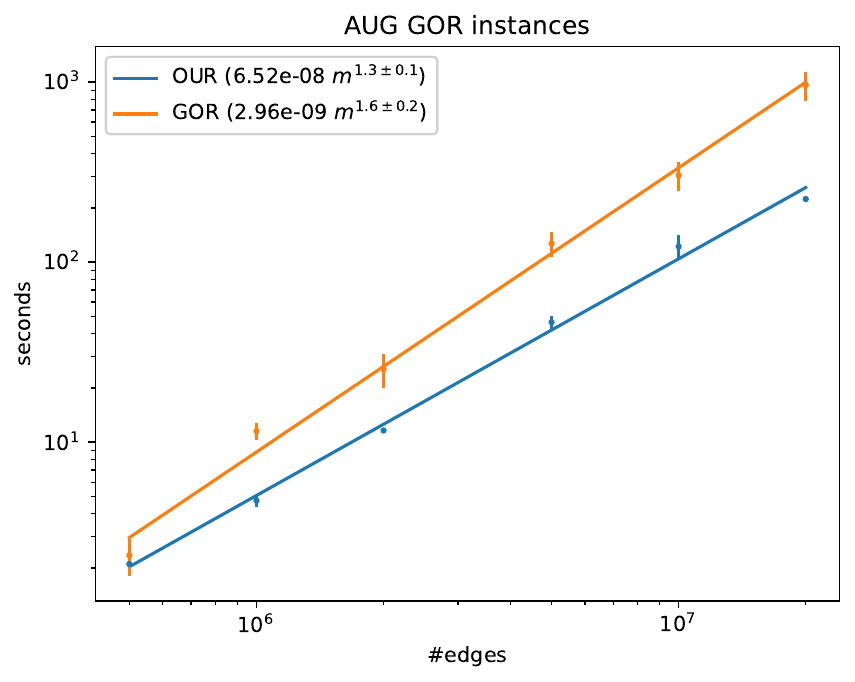}
    \caption{Performance of \BCF vs \GOR on \auggor instances (see Section~\ref{sec:data}).
    }
    \label{fig:gor-bcf-vs-gor}
    \end{subfigure}
    \hfill
    \begin{subfigure}[t]{.47\linewidth}
    \includegraphics[width=\linewidth]{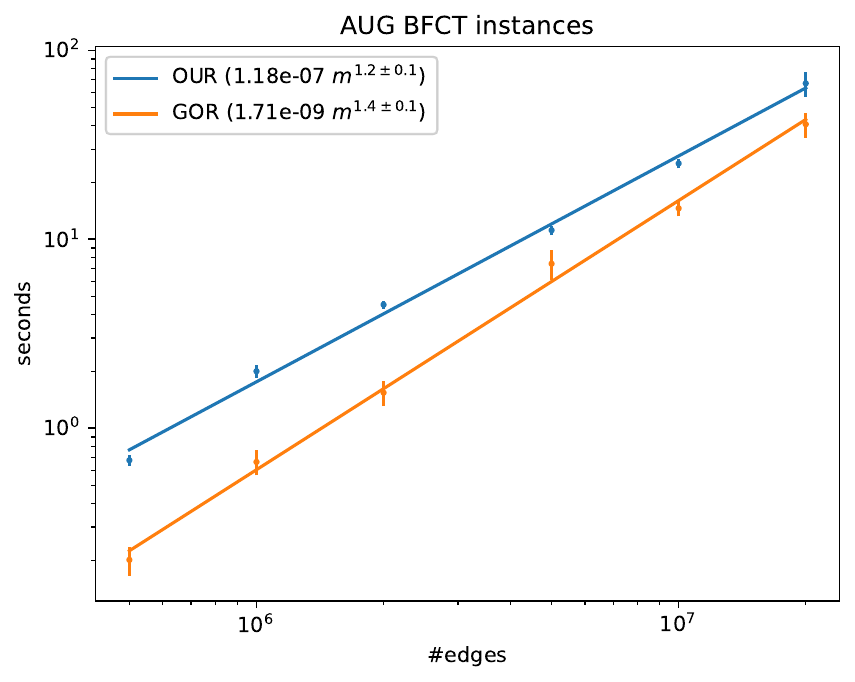}
    \caption{Performance of \BCF vs \GOR on \augbfct instances (see Section~\ref{sec:data}).}
    \label{fig:bfct-bcf-vs-gor}
    \end{subfigure}
    
    \begin{subfigure}[t]{.47\linewidth}
    \includegraphics[width=\linewidth]{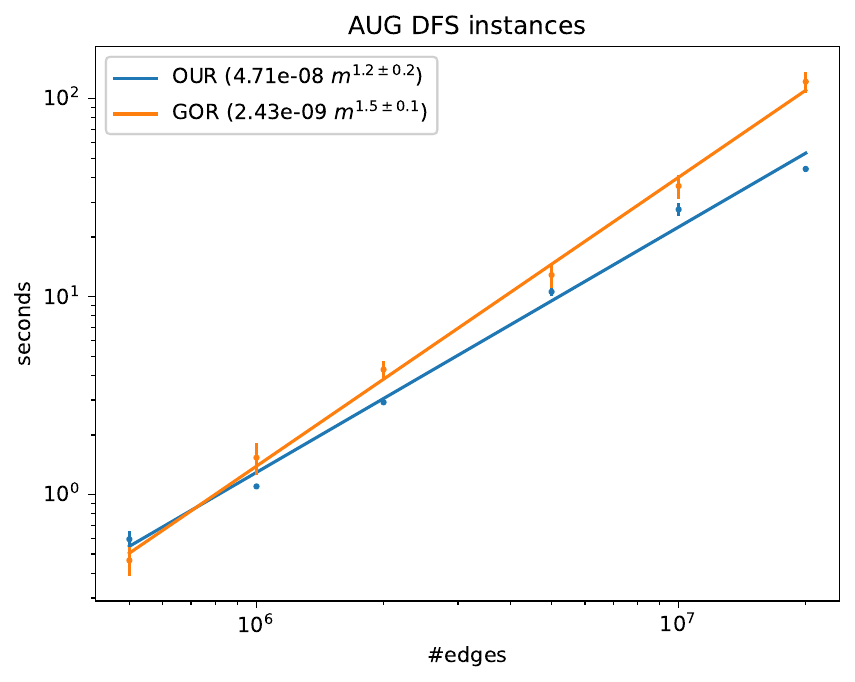}
    \caption{Performance of \BCF vs \GOR on \augdfs instances (see Section~\ref{sec:data}).}
    \label{fig:dfs-bcf-vs-gor}
    \end{subfigure}
    \hfill
    \begin{subfigure}[t]{.47\linewidth}
    \includegraphics[width=\linewidth]{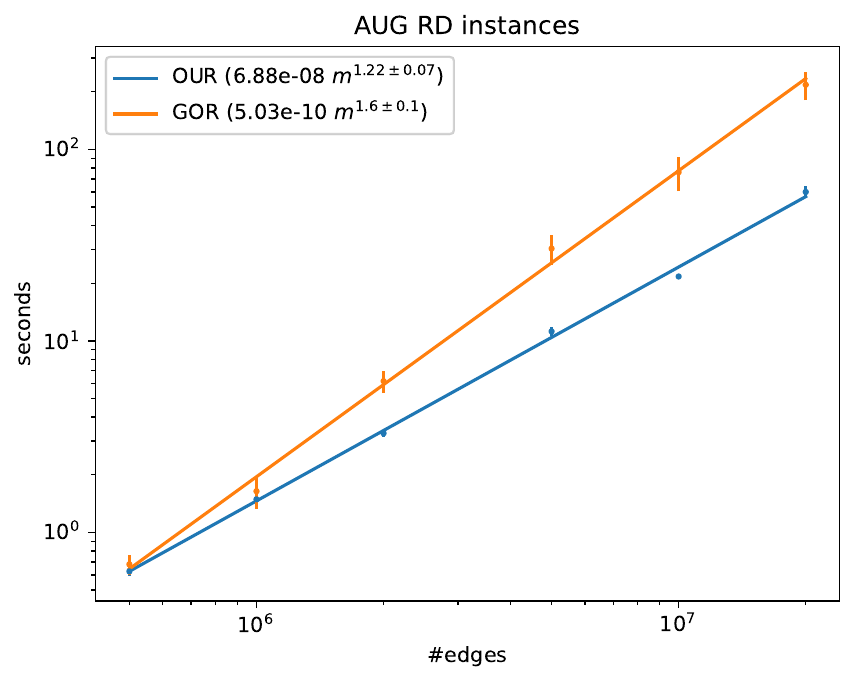}
    \caption{Performance of \BCF vs \GOR on \augrdone instances (see Section~\ref{sec:data}).}
    \label{fig:rd1-bcf-vs-gor}
    \end{subfigure}
    
    \begin{subfigure}[t]{.47\linewidth}
    \includegraphics[width=\linewidth]{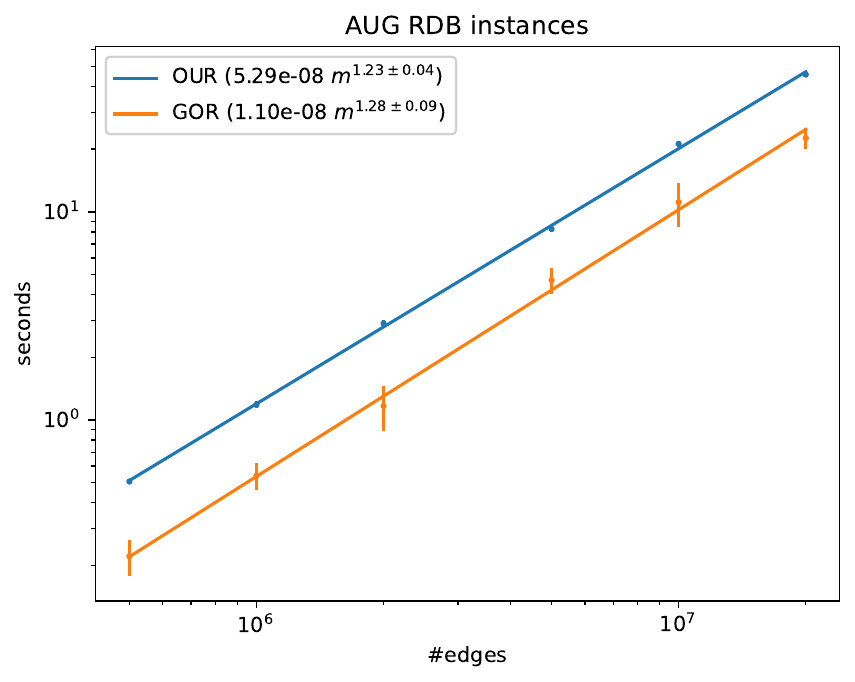}
    \caption{Performance of \BCF vs \GOR on \augrdtwo instances (see Section~\ref{sec:data}).}
    \label{fig:rd2-bcf-vs-gor}
    \end{subfigure}
    \hfill
    \begin{subfigure}[t]{.47\linewidth}
    \includegraphics[width=\linewidth]{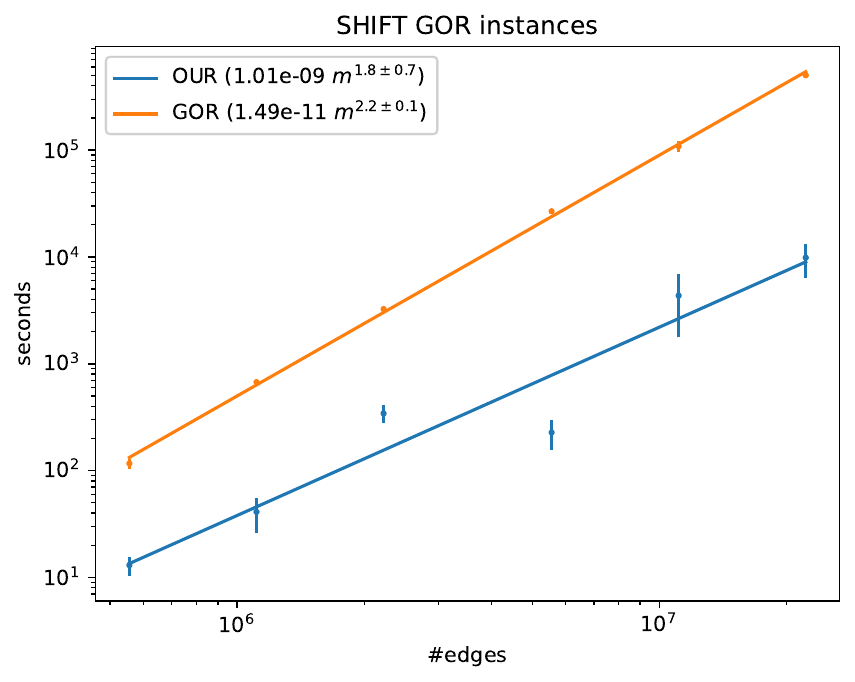}
    \caption{Performance of \BCF vs \GOR on SHIFTED \auggor instances (see Section~\ref{sec:data}).
    }
    \label{fig:shift-bcf-vs-gor}
    \end{subfigure}

    \caption{Our main experiments comparing the running time of \BCF and \GOR.}
    \label{fig:main_experiments}
\end{figure*}

\begin{figure}[t]
    \centering
    \includegraphics[width=0.5\linewidth]{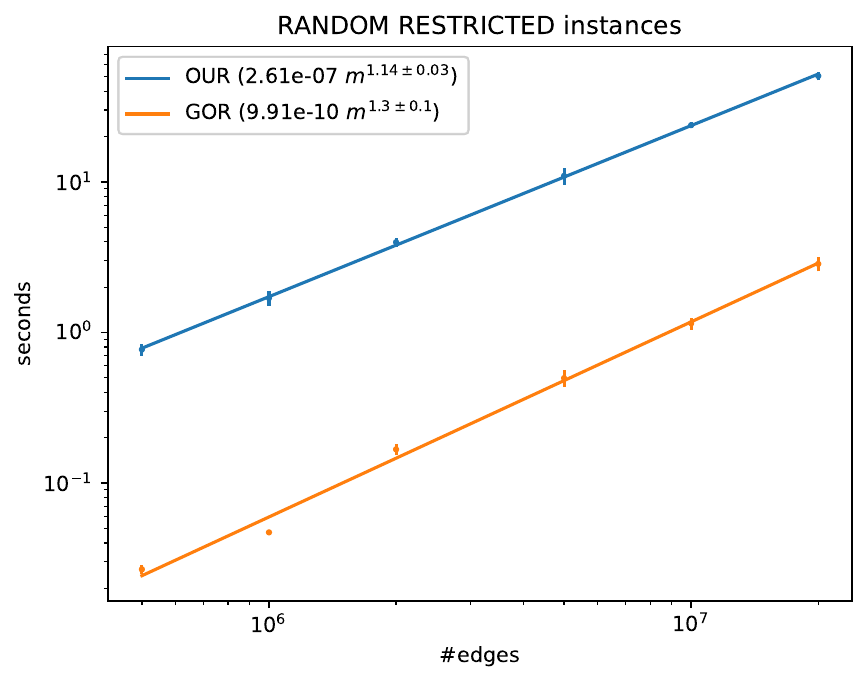}
    \caption{Performance of \BCF vs \GOR on \rand instances.
    }
    \label{fig:rand-bcf-vs-gor}
\end{figure}

\subsection{Comparison} \label{subsec:comparison}
In this section we compare the performance of \BCF and \GOR on several instances. The size of the instances that we use ranges from $5 \cdot 10^5$ edges to $2 \cdot 10^7$ edges.
We consider the total running times but also the scaling behavior of both algorithms.

Let us first consider the behavior of \BCF and \GOR on the \rand instances, see Figure~\ref{fig:rand-bcf-vs-gor}.
The scaling behavior of \BCF is only slightly better than \GOR, however \BCF is consistently slower by a significant factor.
We note, however, that the overall running times for both algorithms on these instances are low compared to our other experiments on graphs with similar sizes.
Let us now consider the DIMACS instance.
We compare the running time of \BCF and \GOR on the unmodified instance, which we refer to as \usa{0}, and on the instance with a potential shift $W$ (see Section~\ref{sec:data}) of $1$, $10$, and $100$ (\usa{1}, \usa{10}, \usa{100}, respectively). On \usa{0} \BCF is $5.4$ time slower. On \usa{1} and \usa{10} \BCF is $7.1$ times slower, and on \usa{100} \BCF is $6.1$ time slower.

\begin{table*}[t]
    \caption{\BCF vs \GOR on \bad instances with $10^{6}$ edges. Time is expressed in seconds.}
    \centering
\begin{tabular}{cccccc}
\toprule
     &                   \badbfct &                    \baddfs &                  \badgor &                  \badrdone &                \badrdtwo \\
\midrule
\BCF &            $1.14 \pm 0.03$ &            $0.71 \pm 0.03$ & $\mathbf{0.95 \pm 0.02}$ &            $0.90 \pm 0.03$ & $\mathbf{0.87 \pm 0.02}$ \\
\GOR & $\mathbf{0.014 \pm 0.003}$ & $\mathbf{0.016 \pm 0.002}$ &             $720 \pm 20$ & $\mathbf{0.028 \pm 0.002}$ &              $260 \pm 3$ \\
\bottomrule
\end{tabular}
    \label{tab:bad-bcf-vs-gor}
\end{table*}

We now consider experiments on the \bad instances described in Section~\ref{sec:bad}.
As discussed in Section~\ref{sec:bad_aug}, a comparison on these instances is misleading as our implementation explicitly exploits the DAG structure while \GOR does not.
However, it is still enlightening to see for which instances \GOR performs poorly.
We show the results of these experiments in Table~\ref{tab:bad-bcf-vs-gor}.
We can see that \BCF has a stable running time on these instances, as expected.
Furthermore, we can clearly see the running time overhead of the data structures of \BCF over the more lightweight \GOR.
However, the instances \badgor and \badrdtwo are clearly pathological for \GOR and it takes excessive time to solve these two instances, while the running time of \BCF is similar to the one for other instances.
We believe that an approach similar to ours that reduces the overhead induced by the data structures could become competitive to \GOR on the instances where it is faster.

We now perform experiments on the \aug instances, see Figures~\ref{fig:gor-bcf-vs-gor}~to~\ref{fig:rd2-bcf-vs-gor}.
First, note that on all \aug instances the scaling behavior of \BCF is better than the scaling behavior of \GOR.
In particular, the regression results for the running times of \BCF on all of the restricted instances (i.e., Figures~\ref{fig:bfct-bcf-vs-gor}~to~\ref{fig:rd2-bcf-vs-gor}) are compatible with the theoretically expected near-linear running time:
As described in Section~\ref{subsec:ouralg}, we could expect a running time of $O((m+n \log n) \log^2 n)$, which is $O(n \log^3 n)$ for $m \in O(n)$.
The same regression that we use for the running times applied to $x \log^3(x)$ for sampled $x \in [5 \cdot 10^5, 2 \cdot 10^7]$ leads to a polynomial factor of $x^{1.20 \pm 0.01}$, which is very close to the experimental results in the restricted cases.
On the non-restricted augmented instance \auggor, the running time of \BCF is significantly worse than on the restricted instances, even though it is still significantly better than the running time of \GOR.
When it comes to absolute running times, our experiments show a diverse picture.
On the \auggor and \augrdone instances, \BCF is always faster than \GOR, see Figures~\ref{fig:gor-bcf-vs-gor}~and~\ref{fig:rd1-bcf-vs-gor}.
However, on the \augbfct and \augrdtwo instances, \GOR is always faster than \BCF, see Figures~\ref{fig:bfct-bcf-vs-gor}~and~\ref{fig:rd2-bcf-vs-gor}.
On \augdfs, \BCF is slower on the smallest instance and consistently faster on all the other sizes, see Figure~\ref{fig:dfs-bcf-vs-gor}.
Summarizing, we want to stress that while \GOR indeed is very fast on some instances, it also exhibits fluctuations of multiple orders of magnitudes on graphs of similar size due its adaptiveness; we do not see such fluctuations in our implementation when running it on restricted instances.

We now consider the experiments on the \shiftgor instance, see Figure~\ref{fig:shift-bcf-vs-gor}.
This instance is particularly hard for \GOR, as it highlights its quadratic nature.
Despite being far from restricted, \BCF performs consistently better.
\BCF outperforms \GOR by almost two orders of magnitudes on some of these non-trivial instances.
It is to notice, however, that the running time of \BCF is strongly unstable, probably because of the strong non-restricted nature of this instance.
Each point in the plot in Figure~\ref{fig:shift-bcf-vs-gor} is the average of the running time of $5$ different random generations of an instance of that size (and only one run per instance). This only marginally smoothed the unstable trend.

\bibliography{main}

\appendix

\section{Instances} \label{sec:app_instances}
In the following we present a comprehensive overview of all the different types of instances we perform experiments on.

\subsubsection*{Augmented \bad instances} \label{sec:app_bad_aug}
The \AUG instances are obtained from the \BAD instances performing two crucial modifications:
\begin{itemize}
    \item We randomly permute the vertex labels, as for many of the instances, the hardness for previous algorithms relies on the specific vertex ordering.
    \item We augment the graphs with 5 times as many edges as the original graph, which are selected uniformly at random using rejection sampling.
    We choose a constant factor as we aim to perform experiments on graphs with many vertices.
    Each of these new edges is assigned a large positive weight such that the minimum cycle mean remains at least 1, as we want to obtain restricted instances (see \Cref{sec:prelims}).
    As the initial graphs are DAGs, this weight is easy to determine: for any path $P$ in $G$, we have to ensure that for the augmenting edge $e$ that closes this path, we have
    \[
    w(P \circ \{e\}) \geq |P|+1 \iff w(e) \geq |P| + 1 - w(P),
    \]
    and that all new edge weights are larger than the edge weights of the original graph.
    We simply set the edge weight of all augmenting edges to a single pre-determined value that ensures the above properties.
    Note that augmenting the graphs with these edges breaks the DAG structure while maintaining the shortest path structure, as desired.
\end{itemize}

Note that the above modifications only ensure that the minimum cycle mean is at least 1, but we did not argue whether the instances fulfill the second condition of restrictedness (see \Cref{sec:prelims}), i.e., whether all edge weights are at least $-1$.
Most of the instances are already restricted or require minimal modifications to be transformed into restricted instances.
\badbfct and \baddfs are restricted without modifications.
The instances \badrdone and \badrdtwo are originally not restricted, as both contain weight $-2$ edges.
However, since any weight $-2$ edge is always preceded by a degree-2 vertex with only an incoming edge of weight $0$, we set both of these edge weights to $-1$.
After this modification both instances are restricted.
The \badgor instance is not restricted, however, the reason it is not restricted is only due to a single edge having weight $-3 \frac{n - 1}{2}$. We keep the instance unrestricted\footnote{However, note that we still ensure that the minimum cycle mean is at least 1 after augmentation.} as there is no straight-forward way to make it restricted, and it leads to insightful experimental results.

\subsubsection*{\shiftgor}
We extract the \shiftgor instances as follows:
\begin{itemize}
    \item Consider the largest SCC in the first recursion step of \Cref{alg:restricted_bcf}, and consider the potentials before the execution of \LazyDijkstra at the end of \Cref{alg:restricted_bcf}.
    \item We add an explicit super source connected with zero weight edges to all vertices.
    \item We then change the edge weights according to the potentials, i.e., $w(u,v) \gets w(u,v) + \phi(u) - \phi(v)$.
\end{itemize}
Since the size of the SCCs depends on the allocation of the augmented edges during the generation of the \auggor instance, the size of the \shiftgor instances may vary slightly.
Furthermore, note that the \shiftgor instances are not restricted as they contain many edges with weight smaller than $-1$.
Despite the non-restricted nature of these instances, we show that \BCF outperforms \GOR on them (see Section~\ref{subsec:comparison}).

\subsubsection*{\rand instances}
The generation process of \rand instances is as follows.
We first generate the structure of the graph by inserting $6n$ edges uniformly at random by rejection sampling.
We set all edge weights to $2$ initially, implying an initial minimum cycle mean of $2$.
We then iteratively run Dijkstra's algorithm on yet unvisited vertices always choosing a random one as source.
This procedure partitions the vertices into shortest path trees.
We perform a potential shift by the distances, resulting in zero weight edges after potential application.
We set the weight of the edges connecting any two shortest path trees to $0$.\footnote{This does not introduce negative cycles as these edges are always forward edges topologically.}
Finally, we decrease all the edge weights by $1$.
Now the minimum cycle mean is at least $1$ and the minimum weights are at least $-1$.

\subsubsection*{\USA instances}
This is a graph with roughly 24 million nodes, 58 million edges, and all its edge weights are positive.
To obtain a graph with negative edge weights for our experiments, we run a shortest path algorithm from the vertex with label zero and assign the distances as potentials to the vertices.
Subsequently, we perform a random shift on the potentials by adding a random number in the range $[0, W]$ to the potential of each vertex.
We then apply the potentials to the edge weights.
After the application, the edge weights on the shortest path tree are in the range $[-W, W]$; for $W = 1$ we obtain restricted instances.

\section{BAD Instances} \label{sec:app_bad_instances}

In this section, we provide a description of the instances that we use in our experimental evaluation that are not yet described in Section~\ref{sec:experiments}.
We note that many of these graphs were first described in~\cite{CherkasskyGGTW09} and we add the descriptions of these graphs here for completeness.
In these cases, we closely follow the description of~\cite{CherkasskyGGTW09}.
We refer to the vertices by their indices.

\subsubsection*{\badbfct}
The graph \badbfct contains $n = 4k - 1$ vertices and $m = 5k - 3$ edges, where $k$ is given as a parameter.
The vertices $1$ to $3k - 2$ together with the edges $(i + 1, i)$, $1 \leq i \leq 3k - 3$, form a path $P$. Every third vertex on $P$ is connected to vertex $3k - 1$, that is, we have the edges $(3(i - 1) + 1, 3k - 1)$ for $1 \leq i \leq k$. Finally, vertex $3k - 1$ is connected to vertices $3k$ to $4k - 1$. All edges in this graph have weight $-1$. Figure~\ref{fig:app_bad-bfct} gives an example for $k = 4$.

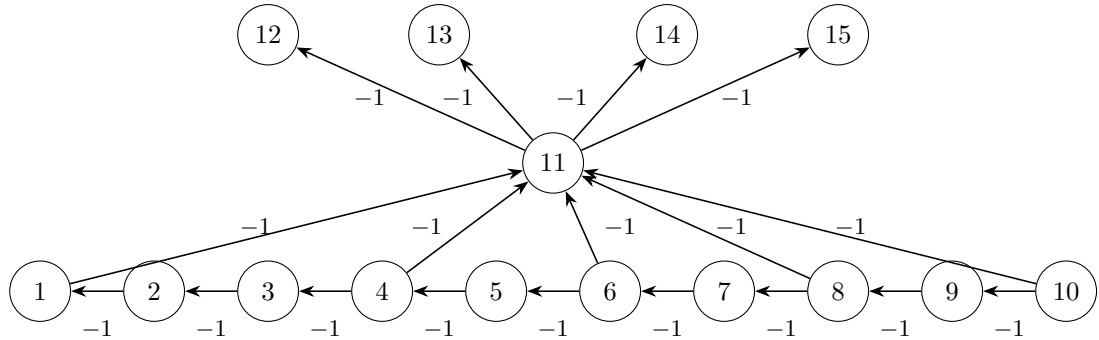
\begin{figure}[htb!]
    \centering
\begin{tikzpicture}[
    node distance=1.5cm,
    every node/.style={circle, draw, minimum size=0.8cm},
    arrow/.style={-Stealth, semithick},
    label/.style={font=\small, inner sep=0pt, outer sep=2pt, fill=none, draw=none}
]

\foreach \x in {1,...,10}
    \node (\x) at (\x*1.5-1.5,0) {\x};

\node (11) at (6.75,1.7) {11};
\node (12) at (3,3.4) {12};
\node (13) at (5.25,3.4) {13};
\node (14) at (8.25,3.4) {14};
\node (15) at (10.5,3.4) {15};

\foreach \x [remember=\x as \lastx (initially 1)] in {2,...,10}
    \draw[arrow] (\x) -- (\lastx) node[midway, below=1pt, label] {$-1$};

\draw[arrow] (1) -- (11) node[midway, left=2pt, label] {$-1$};
\draw[arrow] (4) -- (11) node[midway, left=2pt, label] {$-1$};
\draw[arrow] (11) -- (12) node[midway, left=2pt, label] {$-1$};
\draw[arrow] (11) -- (13) node[midway, left=1pt, label] {$-1$};
\draw[arrow] (11) -- (14) node[midway, left=1pt, label] {$-1$};
\draw[arrow] (11) -- (15) node[midway, right=2pt, label] {$-1$};
\draw[arrow] (6) -- (11) node[midway, right=1pt, label] {$-1$};
\draw[arrow] (8) -- (11) node[midway, right=0pt, label] {$-1$};
\draw[arrow] (10) -- (11) node[midway, right=2pt, label] {$-1$};

\end{tikzpicture}
    \caption{Schematic representation of \badbfct instances with $15$ nodes and $18$ edges.}
    \label{fig:app_bad-bfct}
\end{figure}

\subsubsection*{\badgor}
The graph \badgor consists of a path $P$ of $k$ vertices $1, \dots , k$,
a vertex $k+1$ with $k$ incoming and $k$ outgoing edges, and $k$ vertices $k+2, \dots , 2k+1$. The weights are $w(1, 2) = -3k$, $w(1, k + 1) = -1$, and $w(i, i + 1) = 1$ for $2 \leq i \leq k - 1$. Also, $w(k + 1, k + 1 + i) = -1$ for $1 \leq i \leq k$, and $w(i, k + 1) = 2(k - i)$ for $2 \leq i \leq k$. We have $n = 2k + 1$ and $m = 3k - 1$. Figure~\ref{fig:app_bad-gor} gives an example for $k = 7$.

\begin{figure}[htb!]
    \centering
\begin{tikzpicture}[
    node distance=1.5cm,
    every node/.style={circle, draw, minimum size=0.8cm},
    arrow/.style={-Stealth, semithick},
    label/.style={font=\small, inner sep=0pt, outer sep=2pt, fill=none, draw=none}
]

\node (1) at (0,0) {1};
\node (2) at (2,0) {2};
\node (3) at (4,0) {3};
\node (4) at (6,0) {4};
\node (5) at (8,0) {5};
\node (6) at (10,0) {6};
\node (7) at (12,0) {7};
\node (8) at (6,1.5) {8}; %

\node (9) at (0,3) {9}; %
\node (10) at (2,3) {10};
\node (11) at (4,3) {11};
\node (12) at (6,3) {12};
\node (13) at (8,3) {13};
\node (14) at (10,3) {14};
\node (15) at (12,3) {15};

\draw[arrow] (1) -- (2) node[midway, below=1pt, label] {$-21$};
\draw[arrow] (2) -- (3) node[midway, below=1pt, label] {$1$};
\draw[arrow] (3) -- (4) node[midway, below=1pt, label] {$1$};
\draw[arrow] (4) -- (5) node[midway, below=1pt, label] {$1$};
\draw[arrow] (5) -- (6) node[midway, below=1pt, label] {$1$};
\draw[arrow] (6) -- (7) node[midway, below=1pt, label] {$1$};

\draw[arrow] (1) -- (8) node[midway, left=0.5pt, label] {$-1$};
\draw[arrow] (2) -- (8) node[midway, left=0.5pt, label] {$10$};
\draw[arrow] (3) -- (8) node[midway, left=0.5pt, label] {$8$};
\draw[arrow] (4) -- (8) node[midway, left=0.5pt, label] {$6$};
\draw[arrow] (5) -- (8) node[midway, left=0.5pt, label] {$4$};
\draw[arrow] (6) -- (8) node[midway, left=0.5pt, label] {$2$};
\draw[arrow] (7) -- (8) node[midway, left=0.5pt, label] {$0$};
\draw[arrow] (8) -- (9) node[midway, left=2pt, label] {$-1$};
\draw[arrow] (8) -- (10) node[midway, left=2pt, label] {$-1$};
\draw[arrow] (8) -- (11) node[midway, left=2pt, label] {$-1$};
\draw[arrow] (8) -- (12) node[midway, left=1pt, label] {$-1$};
\draw[arrow] (8) -- (13) node[midway, left=1pt, label] {$-1$};
\draw[arrow] (8) -- (14) node[midway, left=2pt, label] {$-1$};
\draw[arrow] (8) -- (15) node[midway, left=2pt, label] {$-1$};

\end{tikzpicture}
    \caption{Schematic representation of \badgor instances with $15$ nodes and $20$ edges.}
    \label{fig:app_bad-gor}
\end{figure}
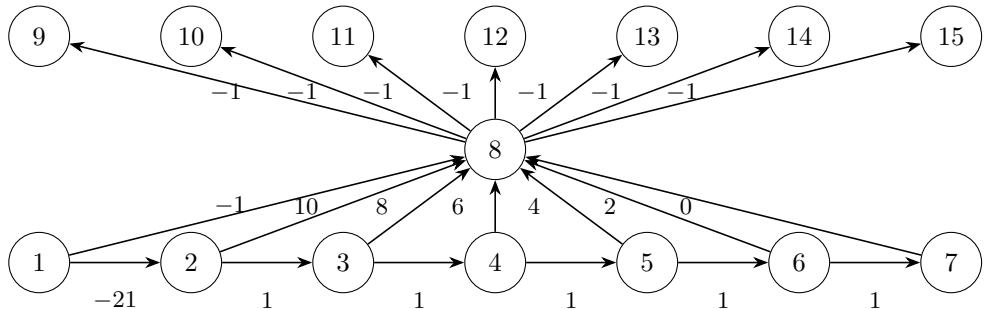

\subsubsection*{\badrdone}
The graph \badrdone consists of $k$ edges $(x_i, y_i) \coloneqq (2i - 1, 2i)$ for
$1 \leq i \leq k$, together with the edges $(x_i, x_i+1)$ and $( y_i, x_i+1)$ for $1 \leq i < k$. We set $w(x_i, y_i) = 0$, $w(y_i, x_i+1) = -2$, and $w(x_i, x_i+1) = -1$. Figure~\ref{fig:app_bad-rd1} gives an example for $k = 5$.

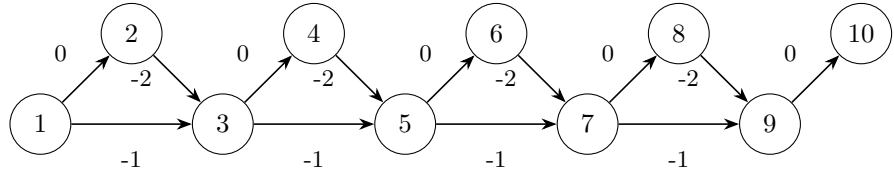
\begin{figure}[htb!]
    \centering
\begin{tikzpicture}[
node distance=1.5cm,
every node/.style={circle, draw, minimum size=0.8cm},
arrow/.style={-Stealth, semithick},
label/.style={font=\small, inner sep=0pt, outer sep=2pt, fill=none, draw=none}
]
\foreach \x in {1,3,5,7,9}
\node (\x) at (\x*1.2-1.2,0) {\x};

\foreach \x in {2,4,6,8,10}
\node (\x) at (\x*1.2-1.2,1.2) {\x};

\foreach \x [remember=\x as \lastx (initially 1)] in {3,5,7,9}
\draw[arrow] (\lastx) -- (\x) node[midway, below, label] {-1};

\foreach \x [remember=\x as \lastx (initially 1), evaluate=\lastx as \previousx using int(\lastx+1)] in {3,5,7,9}
\draw[arrow] (\previousx) -- (\x) node[midway, left, label] {-2};

\foreach \x [remember=\x as \lastx (initially 0), evaluate=\lastx as \previousx using int(\lastx+1)] in {2,4,6,8,10}
\draw[arrow] (\previousx) -- (\x) node[midway, above left, label] {0};

\end{tikzpicture}
    \caption{Schematic representation of \badrdone instances with $10$ nodes and $13$ edges.
    }
    \label{fig:app_bad-rd1}
\end{figure}

\subsubsection*{\badrdtwo}
The graph \badrdtwo is obtained by modifying \badrdone. We connect each even vertex $y_i$ to a new vertex $2k+1$ and connect $2k+1$ to $k$ new vertices $2k+2, \dots , 3k+1$. All new edges have weight $-1$. The new graph has $n = 3k+1$ vertices and $m = 5k-2$ edges. Figure~\ref{fig:app_bad-rd2} gives an example for $k = 5$.

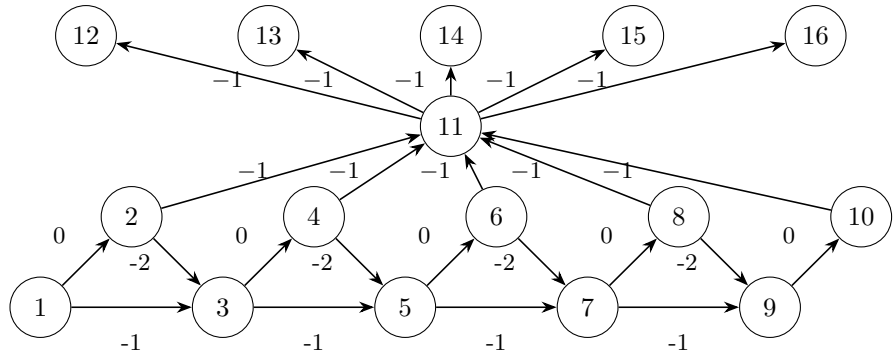
\begin{figure}[htp]
    \centering

\begin{tikzpicture}[
    node distance=1.5cm,
    every node/.style={circle, draw, minimum size=0.8cm},
    arrow/.style={-Stealth, semithick},
    label/.style={font=\small, inner sep=0pt, outer sep=2pt, fill=none, draw=none}
]

\foreach \x in {1,3,5,7,9}
\node (\x) at (\x*1.2-1.2,0) {\x};

\foreach \x in {2,4,6,8,10}
\node (\x) at (\x*1.2-1.2,1.2) {\x};

\node (11) at (5.4,2.4) {11};
\node (12) at (0.6,3.6) {12};
\node (13) at (3,3.6) {13};
\node (14) at (5.4,3.6) {14};
\node (15) at (7.8,3.6) {15};
\node (16) at (10.2,3.6) {16};

\foreach \x [remember=\x as \lastx (initially 1)] in {3,5,7,9}
\draw[arrow] (\lastx) -- (\x) node[midway, below, label] {-1};

\foreach \x [remember=\x as \lastx (initially 1), evaluate=\lastx as \previousx using int(\lastx+1)] in {3,5,7,9}
\draw[arrow] (\previousx) -- (\x) node[midway, left=0.5pt, label] {-2};

\foreach \x [remember=\x as \lastx (initially 0), evaluate=\lastx as \previousx using int(\lastx+1)] in {2,4,6,8,10}
\draw[arrow] (\previousx) -- (\x) node[midway, above left=0.5pt, label] {0};

\foreach \x in {2,4,6,8,10}
    \draw[arrow] (\x) -- (11) node[midway, left=1pt, label] {$-1$};

\foreach \x in {12,...,16}
    \draw[arrow] (11) -- (\x) node[midway, left=2pt, label] {$-1$};

\end{tikzpicture}

    \caption{Schematic representation of \badrdtwo instances with $16$ nodes and $23$ edges.}
    \label{fig:app_bad-rd2}
\end{figure}

\subsubsection*{\baddfs}
The graph \baddfs is also obtained by modifying \badrdone. We add the edges $(y_i, y_i+1)$ for $1 \leq i < k$. We set all edge weights to $-1$. We also relabeled nodes such that the index of the nodes below are all smaller than the indices above.
Figure~\ref{fig:app_bad-dfs} gives an example for $k = 5$.

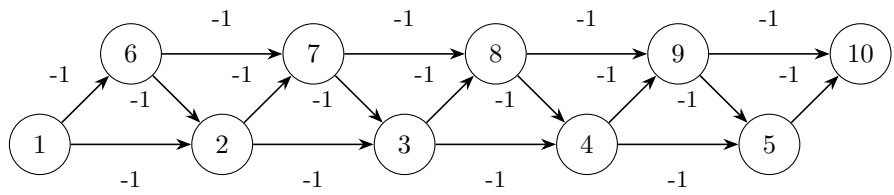
\begin{figure}[htp]
    \centering
\begin{tikzpicture}[
node distance=1.5cm,
every node/.style={circle, draw, minimum size=0.8cm},
arrow/.style={-Stealth, semithick},
label/.style={font=\small, inner sep=0pt, outer sep=2pt, fill=none, draw=none}
]
\foreach \x in {1,2,3,4,5}
\node (\x) at (\x*2.4-2.4,0) {\x};

\foreach \x in {6,7,8,9,10}
\node (\x) at (\x*2.4-13.2,1.2) {\x};

\foreach \x [remember=\x as \lastx (initially 1)] in {2,3,4,5}
\draw[arrow] (\lastx) -- (\x) node[midway, below, label] {-1};

\foreach \x [evaluate=\x as \previousx using int(\x+4)] in {2,3,4,5}
\draw[arrow] (\previousx) -- (\x) node[midway, left, label] {-1};

\foreach \x [evaluate=\x as \nextx using int(\x+5)] in {1,2,3,4,5}
\draw[arrow] (\x) -- (\nextx) node[midway, above left, label] {-1};

\foreach \x [remember=\x as \lastx (initially 6)] in {7,8,9,10}
\draw[arrow] (\lastx) -- (\x) node[midway, above, label] {-1};

\end{tikzpicture}
    \caption{Schematic representation of \baddfs instances with $10$ nodes and $17$ edges.}
    \label{fig:app_bad-dfs}
\end{figure}

\end{document}